\documentclass[letterpaper, 10pt]{article}      % Use this line for a4
\usepackage{graphicx} % for pdf, bitmapped graphics files
\usepackage{times} % assumes new font selection scheme installed
\usepackage{amsmath} % assumes amsmath package installed
\usepackage{amssymb,amsthm}  % assumes amsmath package installed
\usepackage[colorlinks=true]{hyperref}

\newtheorem{theorem}{Theorem}
\newtheorem{lemma}{Lemma}
\newtheorem{remark}{Remark}

\binoppenalty=\maxdimen
\relpenalty=\maxdimen

\newcommand{\ba}{\mathbf{a}}

\newcommand{\bA}{\mathbf{A}}
\newcommand{\bB}{\mathbf{B}}

\newcommand{\bE}{\mathbf{E}}
\newcommand{\bF}{\mathbf{F}}

\newcommand{\bI}{\mathbf{I}}
\newcommand{\bJ}{\mathbf{J}}

\newcommand{\bP}{\mathbf{P}}
\newcommand{\bQ}{\mathbf{Q}}
\newcommand{\bR}{\mathbf{R}}
\newcommand{\bS}{\mathbf{S}}

\newcommand{\bV}{\mathbf{V}}
\newcommand{\bW}{\mathbf{W}}

\newcommand{\cD}{\mathcal{D}}

\newcommand{\ket}[1]{\left|#1\right\rangle}
\newcommand{\bra}[1]{\left\langle #1\right|}
\newcommand{\bket}[1]{\left\langle #1 \right\rangle}

\newcommand{\dotex}{\frac{d}{dt}}

\newcommand{\Tr}[1]{\rm{Tr}\left(#1\right)}

\title{\LARGE \bf
Exponential convergence  of a dissipative quantum system towards finite-energy  grid states of an oscillator
}

\author{Lev-Arcady Sellem\thanks{Laboratoire  de Physique de l’Ecole normale sup\'{e}rieure, Mines Paris-PSL, Inria,  ENS-PSL, Universit\'{e} PSL, CNRS, Sorbonne Universit\'{e},  Paris, France.}
 \and Philippe Campagne-Ibarcq$^*$ \and Mazyar Mirrahimi$^*$ \and Alain Sarlette$^*$ \and Pierre Rouchon$^*$ \thanks{Corresponding author: pierre.rouchon@minesparis.psl.eu }%
}
\begin{document}

\maketitle
\thispagestyle{empty}
\pagestyle{empty}

%%%%%%%%%%%%%%%%%%%%%%%%%%%%%%%%%%%%%%%%%%%%%%%%%%%%%%%%%%%%%%%%%%%%%%%%%%%%%%%%
\begin{abstract}
Based on the stabilizer formalism underlying Quantum Error Correction (QEC),  the design of an original  Lindblad master equation for the density operator of a quantum harmonic oscillator is proposed. This   Lindblad dynamics stabilizes exactly the  finite-energy grid states introduced in 2001 by  Gottesman, Kitaev and Preskill  for  quantum computation.  Stabilization  results from  an exponential Lyapunov function  with an explicit lower-bound on the convergence rate. Numerical simulations indicate the potential interest of such autonomous QEC in presence of non-negligible  photon-losses.

\end{abstract}

%%%%%%%%%%%%%%%%%%%%%%%%%%%%%%%%%%%%%%%%%%%%%%%%%%%%%%%%%%%%%%%%%%%%%%%%%%%%%%%%
\section{Introduction}
Quantum Error Correction (QEC) represents a much sought-after target in the road towards large-scale quantum computations.
Indeed, decoherence affecting early quantum computing platforms
limits their ability to carry out interesting computations.
However, the \emph{threshold theorem}~\cite{nielsen-chang-book} states
that the use of quantum error correcting codes
could allow for arbitrarily long reliable quantum computations,
provided the noise levels affecting the hardware
could be kept below a threshold depending on the considered code.
A major issue for QEC is the huge resource overhead associated with the use of error correcting
codes~\cite{FowlerMariantoniMartinisEtAl2012}
and recent years have seen a growing number of encoding proposals aim at reducing this overhead,
such as the so-called
cat code~\cite{MirrahimiCatComp2014,GuillaudMirrahimiPRX2019},
binomial code~\cite{MichaelSilveriBrierleyEtAl2016}
or GKP code~\cite{GKP-PRA2001}.
In particular, recent experiments in
superconducting circuits~\cite{CampagneIbarcq2020}
and trapped ions~\cite{NeeveGKP2022}
demonstrated the generation
and stabilization
of the finite-energy grid states
underlying the GKP encoding,
sparking a renewed interest for its use for quantum computation (see e.g.\  recent reviews~\cite{TerhalQST2020,GrimsmoPuriGKP2021}).
From a control theoretical perspective, QEC  is a feedback-loop. Usual QEC  is a discrete-time process   based on a static output-feedback where the measured error syndrome (a classical output signal)  indicates which   correcting unitary transformation has to be applied via a specific short time-pulse on the classical control-input signal; in that case the controller is a classical system.
On the other hand, autonomous QEC or reservoir engineering QEC  is  a  continuous-time process where  the controller is  a  dissipative quantum system coupled to the system storing quantum information.
The idea of exploiting quantum dissipation
goes back to optical pumping~\cite{Kastl1967S}.
In~\cite{VerstraeteWolfIgnacioCirac2009}  the potential interest of such dissipation engineering  is  highlighted  for quantum state preparation and  computation.
%%%

Continuous-time stabilization through dissipation engineering
has been experimentally demonstrated for cat codes
(see e.g.\  ~\cite{KirchVLNPGMFGS2013N,LeghtTPKVPSNSHRFSMD2015S,LescanneZaki2019})
and theoretically contemplated for GKP states in~\cite{GirvinPRL20}
where numerical simulations based on  a Lindblad master equation with two   dissipation operators indicate the potential interest of this approach;
however,
the authors did not investigate convergence rates of the proposed dynamics
or energy boundedness along trajectories.
Here  we go further and  propose a set of four dissipation Lindblad operators
exponentially
stabilizing
finite-energy GKP states introduced in~\cite{GKP-PRA2001}.
In  section~\ref{sec:design}, we develop, for a square lattice,
a heuristic   method to design    Lindblad  dynamics
stabilizing  finite-energy GKP states.
Its adaptation to other  lattices such as the hexagonal one is straightforward.
It is inspired by the stabilizer formalism  widely used in quantum error correction (see e.g.\  ~\cite[chapter 10]{nielsen-chang-book}).
Theorem~\ref{thm:Lyapunov} of section~\ref{sec:expo} provides
an  exponential Lyapunov function
for the proposed Lindblad dynamics~\eqref{eq:Lindblad}
with an explicit lower bound on the  convergence rate.
Section~\ref{sec:simu},
devoted to numerical simulations with photon-loss error,
indicates the  interest of  these Lindblad dissipators for autonomous QEC. Possible further developments and issues are gathered in  section~\ref{sec:conclusion}. The detailed and technical  calculations   are in appendix.

%%%%%%%%%%%%%%%%%%%%%%%%%%%%%%%%%%%%%%%%%%%%%%%%%%%%%%%%%%%%%%%%%%%%%%%%%%%%%%%%
\section{Lindblad dissipators derived from infinite-energy stabilizer generators } \label{sec:design}

Set $\eta= 2 \sqrt{\pi}$ and consider the Hermitian phase-space  operators of a quantum harmonic oscillator  $\bQ$ and $\bP$ satisfying  $[\bQ,\bP]=i$.
By Glauber identity $e^{\pm i \eta \bQ}$ and $e^{\pm i \eta \bP}$ commute.
The four commuting operators
$e^{i \eta \bQ}$, $e^{-i \eta \bQ}$, $e^{i \eta \bP}$ and $e^{-i \eta \bP}$
are called the infinite-energy GKP stabilizers
and their common eigenspace associated to the eigenvalue $+1$
is called
the infinite-energy  GKP codespace.
These four stabilizer operators are the independent generators
of the stabilizer group $\{e^{i n \eta \bQ} e^{i m \eta \bP}~|~(n,m)\in\mathbb{Z}^2\}$.

 In the $q$-representation,
 $\bP\equiv - i \frac{d}{dq}$
 hence  $e^{\pm i\eta\bP} \equiv e^{\pm \eta \frac{d}{dq}}$
 corresponds to a constant shift of $\pm \eta$ on $q$.  Thus, $e^{i n \eta \bQ} e^{i m \eta \bP}$ applied on the wave function $\ket{\psi}\equiv (\psi(q))_{q\in\mathbb{R}}$ reads
\begin{equation}
	\label{eq:qrep}
e^{i n \eta \bQ} e^{i m \eta \bP} \ket{\psi} \equiv \left(e^{i n \eta q} \psi(q+m\eta)\right)_{q\in\mathbb{R}}
.
\end{equation}
Solving for the $+1$-eigenstates of~\eqref{eq:qrep},
we find that the infinite-energy GKP codespace is of dimension $2$ and spanned by two Dirac combs, the even comb
$\sum_{k\in\mathbb{Z}} \delta\big(q-2k\tfrac{2\pi}\eta\big)$
located at even multiples of $\tfrac{2\pi}\eta=\sqrt\pi$
and the odd comb
$\sum_{k\in\mathbb{Z}} \delta\big(q-(2k+1)\tfrac{2\pi}\eta\big)$
located at odd multiples of $\tfrac{2\pi}\eta$
($\delta$ stands for the Dirac distribution).

As  in~\cite[section V]{GKP-PRA2001} (see also \cite{GrimsmoPuriGKP2021} for a recent exposure), consider $\bE_\varepsilon=e^{- \frac{\varepsilon}{2} (\bQ^2+\bP^2)}$, a regularizing Hermitian operator  with  $0 <\varepsilon\ll 1$.  In the $q$-representation $\bE_\varepsilon$  corresponds to the  convolution with the Mehler  kernel
\[
K(q,q',\varepsilon)= \tfrac{\exp\left(- \frac{\tanh\varepsilon}{2} (q')^2\right)}{\sqrt{2\pi\sinh\varepsilon}} \exp\left(- \tfrac{(q-q'/\cosh\varepsilon)^2}{2\tanh\varepsilon}\right)
.
\]
Thus, $\bE_\varepsilon \ket{\psi}$ reads  $ \int_{\mathbb{R}} K(q,q',\varepsilon ) \psi(q') \, dq'  $ and applied to the even and odd  combs yields
the following  coherent  superpositions  of Gaussian squeezed states of finite-energy (average photon-number around  $1/(2\varepsilon)$):
\begin{align*}
	\ket{\text{even}_\varepsilon}&\equiv  \left(\sum_k   \tfrac{e^{-  \frac{\pi \tanh\varepsilon}{2}(2k)^2}}{\sqrt{2\pi\sinh\varepsilon}} e^{- \frac{(q-2k\sqrt{\pi}/\cosh\varepsilon)^2}{2\tanh\varepsilon}} \right)_{q\in\mathbb{R}}\\
	\ket{\text{odd}_\varepsilon}&\equiv \left(\sum_k   \tfrac{e^{- \frac{\pi \tanh\varepsilon}{2}(2k+1)^2}}{\sqrt{2\pi\sinh\varepsilon}} e^{- \frac{(q-(2k+1)\sqrt{\pi}/\cosh\varepsilon)^2}{2\tanh\varepsilon}}\right)_{q\in\mathbb{R}}.
\end{align*}
With  $0<\varepsilon \ll 1$ these two finite-energy and smooth  quantum states
approximate generators of the  infinite-energy  GKP codespace.
We introduce an orthonormal basis of their span, defined by
$
\ket{0_\varepsilon} \propto  \ket{\text{even}_\varepsilon}$ and $\ket{1_\varepsilon}\propto \ket{\text{odd}_\varepsilon} -  \tfrac{\bket{\text{even}_\varepsilon|\text{odd}_\varepsilon}}{\bket{\text{even}_\varepsilon|\text{even}_\varepsilon}} \ket{\text{even}_\varepsilon}
.
$
By construction, $\ket{0_\varepsilon}$ and $\ket{1_\varepsilon}$  belong to  the kernel of the following four non-Hermitian operators derived from the infinite-energy stabilizer operators:
\begin{multline*}
         \bV_1 = \bE_\varepsilon \, e^{i \eta \bQ} \, \bE_\varepsilon^{-1}- \bI,~
         \bV_2 = \bE_\varepsilon \, e^{-i \eta \bQ} \, \bE_\varepsilon^{-1}- \bI,
         \\
         \bV_3 = \bE_\varepsilon \, e^{i \eta \bP} \, \bE_\varepsilon^{-1}- \bI,~
         \bV_4 = \bE_\varepsilon \, e^{-i \eta \bP} \, \bE_\varepsilon^{-1}- \bI.
 \end{multline*}
Using \begin{align*}
 & \bE_{\varepsilon} \, \bQ \, \bE_{\varepsilon}^{-1}  = \cosh(\varepsilon) \bQ + i \sinh(\varepsilon) \bP \triangleq \bR
\\
&\bE_{\varepsilon} \, \bP \, \bE_{\varepsilon}^{-1}  = - i \sinh(\varepsilon) \bQ +  \cosh(\varepsilon) \bP \triangleq \bS
\end{align*}
these four operators read
\begin{multline}\label{eq:Vk}
         \bV_1 = e^{i\eta \bR}-\bI,~
         \bV_2 = e^{i\eta \bS}-\bI,~
         \bV_3 = e^{-i\eta \bR}-\bI,~
         \bV_4 = e^{-i\eta \bS}-\bI .
 \end{multline}
Since  $[\bR,\bS]=i$ and $\eta^2=4\pi$,  for any $k,\ell$, $\bV_k$ and $\bV_\ell$ commute.
Then any density operator $\rho$ having his support in $\text{span}\{ \ket{0_\varepsilon},\ket{1_\varepsilon}\}$  is a steady state of the following Lindblad master equation:
\begin{equation}\label{eq:Lindblad}
  \dotex \rho =\sum_{k=1}^{4} \cD_{\bV_k}(\rho)\triangleq \mathcal{L}_\varepsilon(\rho)
\end{equation}
where
$
\cD_{\bV}(\rho) \triangleq \bV\rho \bV^\dag  - (\bV^\dag \bV \rho + \rho \bV^\dag\bV)/2
.
$
Next section  provides a first formal analysis ensuring the exponential  convergence of the above dynamical system towards the finite-energy GKP codespace, i.e.\   towards the set of  density operators $\rho$  with range  in  $\text{span}\{ \ket{0_\varepsilon},\ket{1_\varepsilon}\}$.

%%%%%%%%%%%%%%%%%%%%%%%%%%%%%%%%%%%%%%%%%%%%%%%%%%%%%%%%%%%%%%%%%%%%%%%%%%%%%%%%
\section{Exponential convergence} \label{sec:expo}

The rigorous functional analysis framework  is not addressed here:
calculations are led as if the dimension of the underlying Hilbert space were finite.
The \emph{a priori} estimate we obtain
constitutes a first step towards
a fully rigorous mathematical analysis
that we plan to develop in future publications.

\begin{theorem}\label{thm:Lyapunov}
Consider  $
  \bW= \sum_{k=1}^{4} \bV_k^\dag \bV_k$ where the $\bV_k$ are given by~\eqref{eq:Vk}.
Then for any  time-varying density operator $\rho(t)$ satisfying~\eqref{eq:Lindblad}, we have,   for $\eta=2\sqrt{\pi}$ and  $\varepsilon\in(0,\frac{1}{2\eta}]$:
	\[
\dotex \Tr{\bW\rho(t)} \leq - \kappa(\varepsilon,\eta)  \Tr{\bW\rho(t)}
\]
with  $\kappa(\varepsilon,\eta) >0$ given by
\begin{multline}\label{eq:kappa}
  \kappa(\varepsilon,\eta)=\big(\sinh(\eta^2 s)-\sin(\eta^2c)\big) \big(1- e^{-3\eta^2 s/2}\big)
  \\
  -\big(\cosh(\eta^2 s)-\cos(\eta^2c)\big)  \big(1+ e^{-3\eta^2 s/2}\big)
\end{multline}
where $s=\sinh(2\varepsilon)$ and $c=\cosh(2\varepsilon)$.
For $0<\varepsilon\ll 1$ and $\eta=2\sqrt{\pi}$,  we have $\kappa(\varepsilon,\eta) =2 \eta^4 \varepsilon^2 + O(\varepsilon^3)$.
\end{theorem}

The detailed proof is quite technical. It is  given in appendix~\ref{ap:proof} and relies on Glauber identity.
This theorem implies that,
for any initial density operator $\rho(0)$,
$0 \leq \Tr{\bW\rho(t)} \leq \Tr{\bW\rho(0)} e^{-  \kappa(\varepsilon,\eta) t}$.
Thus, $\lim_{t\mapsto +\infty} \Tr{\bW\rho(t)} =0$.
Since, for all $t \geq 0$,
$\rho(t) \geq 0$ and $\bW\geq 0$,
the support of $\rho(t)$ converges to
$\ker\bW=\text{span}\{\ket{0_\varepsilon},\ket{1_\varepsilon}\}$,
the finite-energy GKP codespace.
Since any operator with support in $\ker\bW$
belongs to $\ker\mathcal{L}_\varepsilon$,
$\rho(t)$ exponentially converges to a steady state of~\eqref{eq:Lindblad}.

Moreover, $\ker\mathcal{L}_\varepsilon$
coincides with  operators having their support  in $\ker\bW$.
Thus, $\ker\mathcal{L}_\varepsilon$ is of real dimension $4$,
spanned by density operators with support on $\text{span}\{\ket{0_\varepsilon},\ket{1_\varepsilon}\}$.
\begin{remark} \label{rmk:sqrt2pi}
  The a priori estimate of theorem~\ref{thm:Lyapunov} is also valid when $\eta=\sqrt{2\pi}$. Then  $\ker\bW$ is  spanned by
a single wave function corresponding to the regularization of the Dirac comb  $\sum_{k\in\mathbb{Z}} \delta\big(q-k\sqrt{2\pi}\big)$ and colinear to
	\[
 \sum_k   \tfrac{e^{- k^2 \pi \tanh\varepsilon}}{\sqrt{2\pi\sinh\varepsilon}} e^{- \frac{(q-k\sqrt{2\pi}/\cosh\varepsilon)^2}{2\tanh\varepsilon}}
.\]
Such  grid states   are certainly to be considered as interesting  resources in metrology  to measure simultaneously  the commuting modular observables derived from  $e^{\pm i\eta\bQ}$ and $e^{\pm i\eta\bP}$ and thus to avoid  the Heisenberg  uncertainty principle attached to measurements  of  $\bQ$ and $\bP$ (see~\cite[chapter V, section 4]{NeumannBook55}).
\end{remark}

Consider
\vspace{-2mm}
\begin{multline*}
  \bS_0 = \ket{0_\varepsilon} \bra{0_\varepsilon} + \ket{1_\varepsilon} \bra{1_\varepsilon}
, ~
\bS_x = \ket{1_\varepsilon} \bra{0_\varepsilon} + \ket{0_\varepsilon} \bra{1_\varepsilon}
,
\\
\bS_y= i \ket{1_\varepsilon} \bra{0_\varepsilon} - i \ket{0_\varepsilon} \bra{1_\varepsilon}
,~
\bS_z =  \ket{0_\varepsilon} \bra{0_\varepsilon} -  \ket{1_\varepsilon} \bra{1_\varepsilon}
.
\end{multline*}
Since $\ker\mathcal{L}_\varepsilon$ is of real dimension $4$, the kernel $\ker\mathcal{L}_\varepsilon^*$ of its adjoint $\mathcal{L}_\varepsilon^*$   for the Frobenius product is also  of dimension $4$. It is  spanned by four independent  Hermitian  invariant  operators,  $\bI$ (conservation of the trace) and
$$
\bJ_\xi =\lim_{t\mapsto +\infty} e^{t \mathcal{L}_\varepsilon^*}(\bS_\xi) ,~ \xi=x,y,z.
$$
Since the spectra  of $\bS_x$, $\bS_y $ and $\bS_z$ are  $\{-1,0,1\}$, the spectra  of
$\bJ_x$, $\bJ_y$ and $\bJ_z$  are inside   $[-1,1]$~\cite{SepulSR2010}.
Then  for any operator $\rho$,  we have (see e.g.\  \cite{AlberJ2014PRA}):
$$
\lim_{t\mapsto +\infty} e^{t \mathcal{L}_\varepsilon}(\rho) = \tfrac{\bS_0 + \Tr{\bJ_x \rho} \bS_x+ \Tr{\bJ_y \rho} \bS_y+ \Tr{\bJ_z \rho} \bS_z}{2}
.
$$
The quantities
$\Tr{\bJ_x \rho}, \, \Tr{\bJ_y \rho}, \, \Tr{\bJ_z \rho}$
can be seen as the Bloch coordinates
of a logical qubit encoded in the  density operator $\rho$,
as they always satisfy
\[(\Tr{\bJ_x \rho})^2+ (\Tr{\bJ_y \rho})^2+(\Tr{\bJ_z \rho})^2 \leq 1.\]

\section{Simulations with photon-loss errors} \label{sec:simu}

\begin{figure}[htbp]
  \centering
  \includegraphics[width=0.45\textwidth]{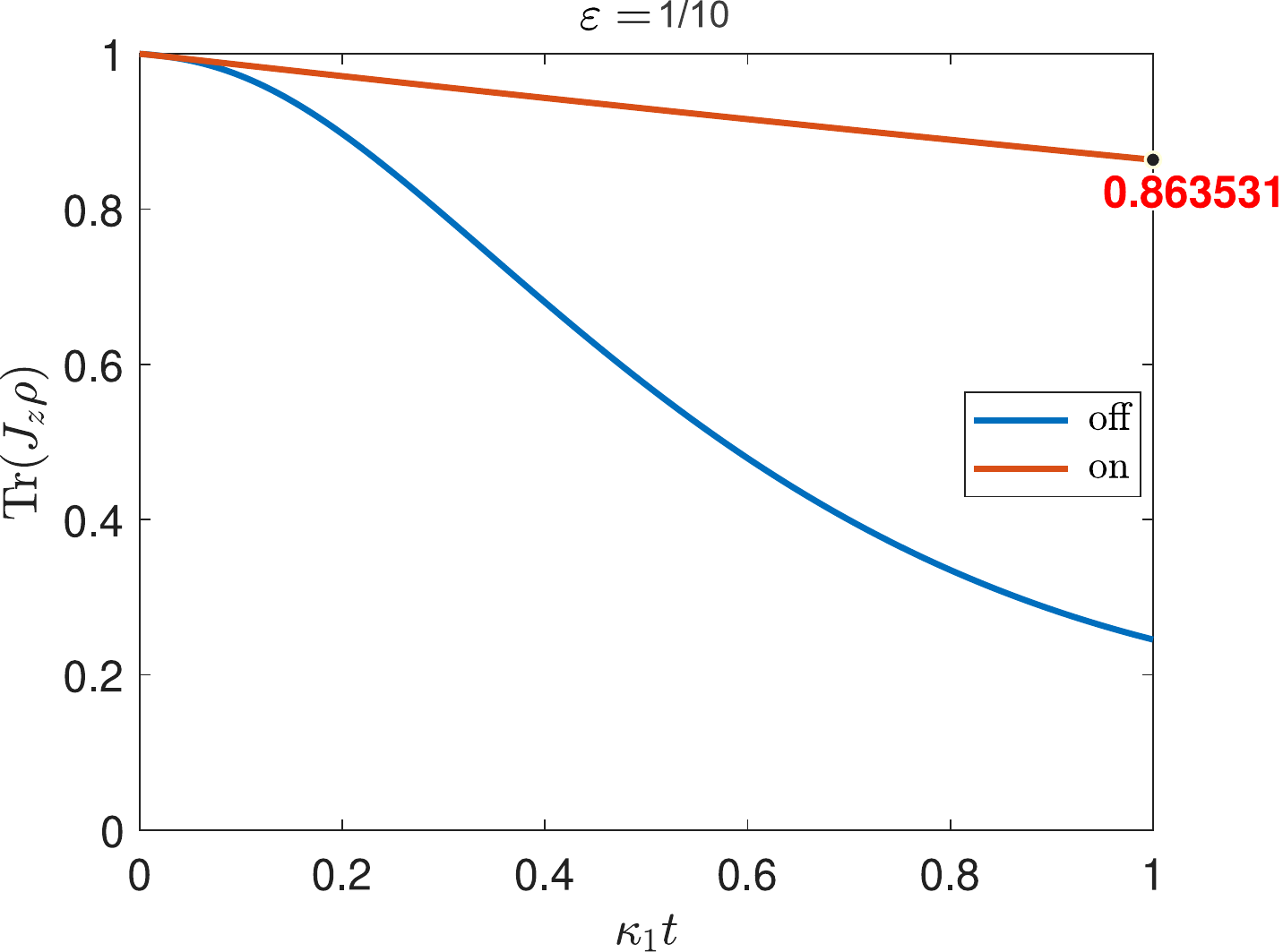} \\
	\vspace{3mm}
  \includegraphics[width=0.45\textwidth]{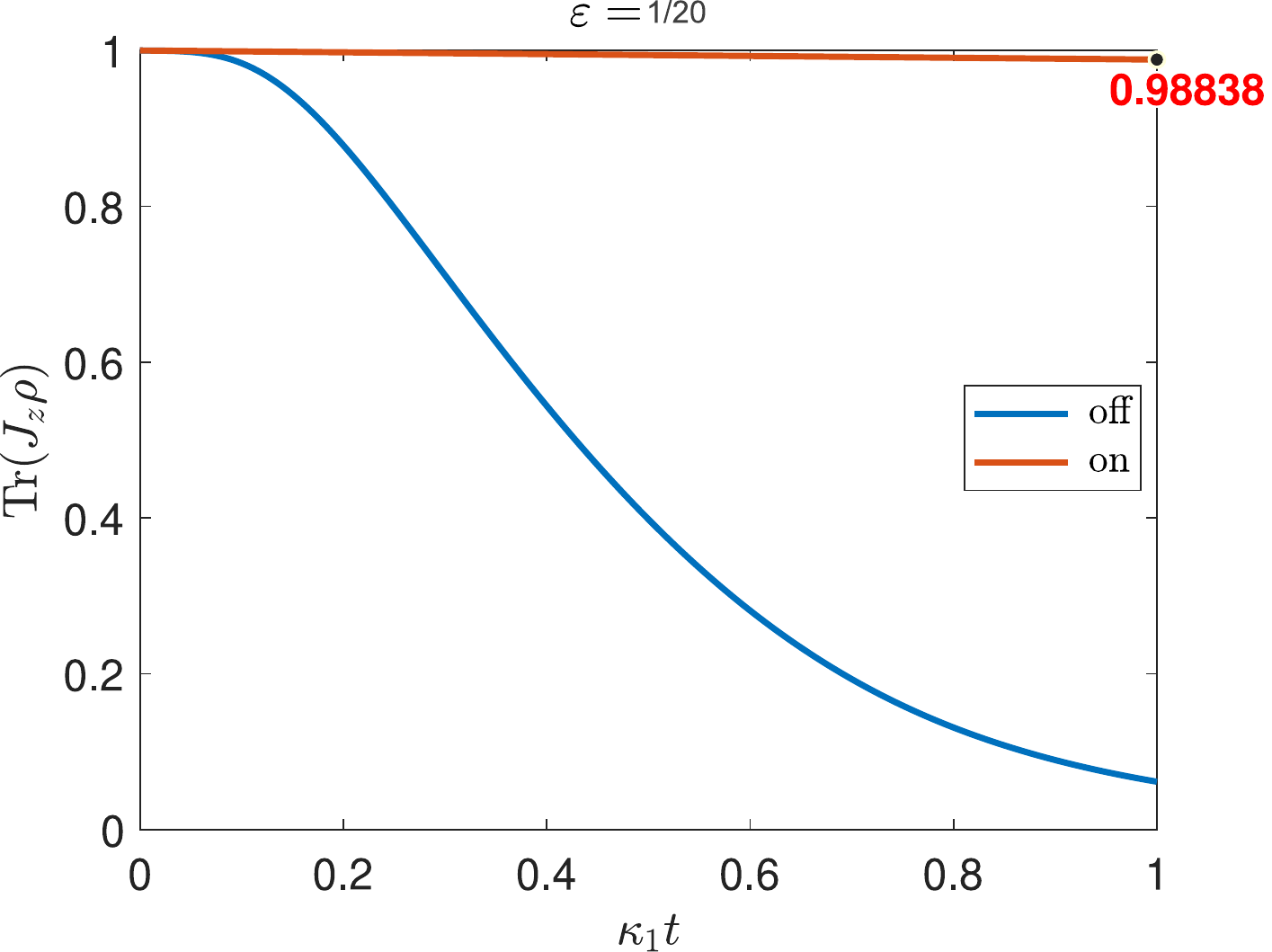} \\
	\vspace{3mm}
  \includegraphics[width=0.45\textwidth]{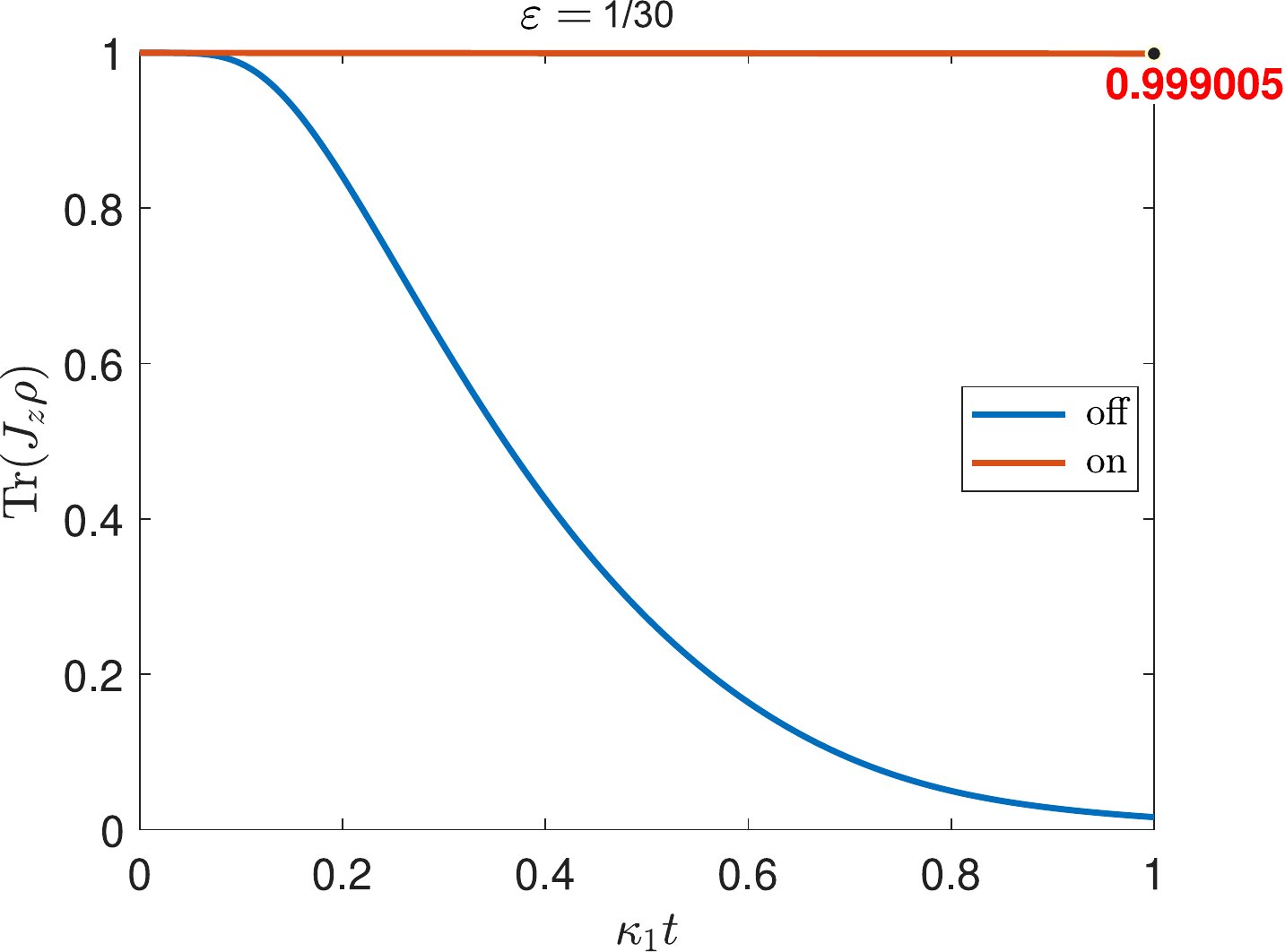}
  \caption{\rm \small Simulations between $t_0=0$ and $t_f=1/\kappa_1$, %
	starting from $\rho(0)=\ket{0_\varepsilon}\bra{0_\varepsilon}$ %
	(implying $\Tr{\bJ_z \rho(0)}=1$) %
	for different value of $\varepsilon$, %
	with (label "on") and without (label "off") %
	the autonomous error correction based on~\eqref{eq:Lindblad}, %
	including   a  non-negligible  error rate  of %
	$\kappa_1=\varepsilon/5$ associated to photon-losses %
	(single-photon life-time of $1/\kappa_1$). %
	The "on" error-rate is estimated as %
	$\kappa_1(1-\Tr{\bJ_z \rho(1/\kappa_1)})$. %
	When $\varepsilon$ decreases, %
	the corresponding decrease of the photon-loss rate $\kappa_1$ %
	is compensated for %
	by the choice of time horizon %
	$t_f = 1/\kappa_1$. %
	We observe that for $\varepsilon=1/10$ (resp. $1/20$ and $1/30$), %
	the "on" error-rate is approximately %
	$1/7$ %
	(resp. $1/80$ and $1/1000$) %
	of %
	the bare "off" error-rate $\kappa_1$.}
\label{fig:Jz}
\end{figure}

The  above formulae  are used in our simulations to compute numerically   $\bS_\xi$ and $\bJ_\xi$ just by  numerical time integration of~\eqref{eq:Lindblad} and of its adjoint.    A Galerkin approximation is used with   Fock subspace $\Big\{\ket{n}\Big\}_{0\leq n \leq n^*}$ where $\ket{n}$ is the state with $n$ photons~\cite{haroche-raimondBook06}.  Since the average number of photons on the finite-energy GKP codespace is around $1/(2\varepsilon)$, $n^*$ has to be much larger than $1/\varepsilon$. We have observed numerically  that  taking $n^*$ around $ 20 /\varepsilon$  is enough since  higher values  do not change the results. On figure~\ref{fig:Jz},  we have performed  simulations for $\varepsilon=1/10$, $1/20$ and $1/30$.
All simulations start   with
$\rho_0=\ket{0_\varepsilon}\bra{0_\varepsilon}$
on the finite-energy GKP codespace,
i.e.\  with logical coordinates
$\Tr{\bJ_x \rho_0}=\Tr{\bJ_y \rho_0}=0$
and
$\Tr{\bJ_z \rho_0}=1$.
All simulations include photon-loss errors at a rate $\kappa_1=\varepsilon/5$  scaled as 10\% of the inverse of the average number of photon in  $\ket{0_\varepsilon}$ and $\ket{1_\varepsilon}$.
The  Lindblad master equations numerically solved are of two kinds:
$$
\dotex \rho = \left\{
                \begin{array}{ll}
                  \mathcal{L}_\varepsilon (\rho) + \frac{\varepsilon}{5} \mathcal{D}_{\ba}(\rho), & \hbox{ curve label  ''on'';} \\
                  \frac{\varepsilon}{5} \mathcal{D}_{\ba}(\rho), & \hbox{curve label ''off'',}
                \end{array}
              \right.
$$
where $\ba=(\bQ+i\bP)/\sqrt{2}$ is the annihilation operator.
We observe a strong suppression of errors in presence of the engineered dissipation $\mathcal L_\varepsilon$.
Other simulations not presented here  with local phase-space operators  (polynomial of low-order in  $\bP$ and $\bQ$)  such as  $\mathcal{D}_{\ba^\dag}$,  $\mathcal{D}_{\bQ}$ and $\mathcal{D}_{\bP}$ instead of  $\mathcal{D}_{\ba}$,  exhibit a similar  strong decrease  of  the decoherence rate when $\varepsilon$ is decreased.

%%%%%%%%%%%%%%%%%%%%%%%%%%%%%%%%%%%%%%%%%%%%%%%%%%%%%%%%%%%%%%%%%%
\section{Concluding remarks}\label{sec:conclusion}

The guarantee of exponential stability
provided by theorem~\ref{thm:Lyapunov},
combined with the
numerically observed
efficient protection
against local errors in phase space, motivates the following issue:
how to physically implement the  autonomous stabilization scheme
attached to the Lindblad master equation~\eqref{eq:Lindblad}?
Quantum superconducting circuits~\cite{CampagneIbarcq2020}
and trapped ions~\cite{NeeveGKP2022}
appear as promising platforms
for this task.

The strong impact of $\varepsilon$ close to $0^+$ on the decoherence rates  is an indication of some exponential behaviour in the protection  against local errors.   This point will be  investigated in
future works.

 Notice the analogy between  the Lyapunov  function $\bW$ in theorem~\ref{thm:Lyapunov}
and the Lyapunov function
$(\ba^r -\alpha^r\bI)^\dag (\ba^r -\alpha^r\bI)$
introduced in~\cite{AzouitSarletteRouchon2016}
for the Lindblad master equation corresponding to multi-photon pumping and cat-qubits~\cite{MirrahimiCatComp2014}:
$\dotex \rho = \mathcal{D}_{\ba^r -\alpha^r\bI}(\rho)$ with
$r\in \mathbb N^*$ and  $\alpha\in\mathbb{C}$.
As already done in~\cite{AzouitSarletteRouchon2016} for cat-qubits, we expect to provide in  forthcoming publications  a fully rigorous  and  functional analysis proof of well-posedness and exponential convergence   of the infinite-dimensional initial-value  problem~\eqref{eq:Lindblad}.

%%%%%%%%%%%%%%%%%%%%%%%%%%%%%%%%%%%%%%%%%%%%%%%%%%%%%%%%%%%%%%%%%%%%%%%%%%%%%%%
\paragraph{Acknowledgments. }
This project has received funding from the European Research Council (ERC) under the European Union’s Horizon 2020 research and innovation programme (grant agreement No. [884762]).

\appendix

\section{Proof of theorem~\ref{thm:Lyapunov}} \label{ap:proof}

Elementary numerical computations  combined with the asymptotics $2\eta^4 \varepsilon^2$ around $0^+$ ensure  that $\kappa >0$  when
$\eta=2\sqrt{\pi}$ and  $\varepsilon\in(0,\frac{1}{2\eta}]$.

Formally  $\dotex \Tr{\bW\rho(t)}= \Tr{\rho(t)\sum_{k=1}^{4} \cD^*_{\bV_k}(\bW)}$ where  the adjoints $\cD^*_{\bV_k}$ of the super operators $\cD_{\bV_k}$  are given by
\begin{align*}
	\cD^*_{\bV_k}(\bW) &\triangleq \bV_k^\dag \bW \bV_k  - (\bV_k^\dag \bV_k \bW + \bW \bV_k^\dag\bV_k)/2
\\
	&\equiv  (\bV_k^\dag [\bW,\bV_k] + [\bV_k^\dag,\bW] \bV_k)/2
.
\end{align*}
Since  for any $k$ and $\ell$, $[\bV_\ell,\bV_k]=0$, we have $[\bV^\dag_\ell\bV_\ell,\bV_k]=[\bV^\dag_\ell, \bV_k] \bV_\ell $. Thus
$$
\dotex \bW \triangleq \sum_{k=1}^{4} \cD^*_{\bV_k}(\bW) =  \sum_{k,\ell} \bV_k^\dag [\bV^\dag_\ell, \bV_k] \bV_\ell
.
$$
Let us introduce the notation
\[ \bR_1 = \bR, \,
	\bR_2 = \bS, \,
	\bR_3 = -\bR, \,
	\bR_4 = -\bS\]
such that
\[ \bV_k = e^{i\eta\bR_k} - \bI.\]
Exploiting Glauber identity, we get
$$
e^{i \eta \bA} e^{ i\eta \bB} = e^{-\tfrac{\eta^2}{2} [\bA,\bB]} \, e^{i\eta (\bA+\bB)}
$$
for operators $\bA,\bB$ such that $[\bA,[\bA,\bB]]=[\bB,[\bA,\bB]]=0$,
from which
 \[
	 \left[ e^{i\eta \bA}, \, e^{i\eta \bB} \right]
	 = e^{i\eta \bB} \,
		\left( e^{-\eta^2[\bA,\bB]}-\bI \right) \,
	   e^{i\eta \bA}.
 \]
With $\bA,\bB\in\{ \bR,\bS,\bR^\dag,\bS^\dag\}$ and commutations
\[
[\bR,\bR^\dag]=[\bS,\bS^\dag]= \sinh(2\varepsilon)\bI
\text{ and }
[\bR,\bS^\dag]= i \cosh(2\varepsilon) \bI,
\]
we get
\begin{align*}
	\frac d{dt} \bW
	&=\sum_{k,\ell} \bV_k^\dag [\bV_\ell^\dag,\bV_k]\bV_\ell \\
	&= \sum_{k,\ell} \bV_k^\dag \, [e^{-i\eta\bR_\ell^\dag}, \, e^{i\eta\bR_k}] \, \bV_\ell\\
	&= \sum_{k,\ell} \bV_k^\dag \,
		e^{i\eta\bR_k} \,
		\left( e^{\eta^2\left[\bR_\ell^\dag,\bR_k\right]}-\bI \right)\,
		e^{-i\eta\bR_\ell^\dag}
		\bV_\ell\\
	&= \sum_{k,\ell} \bW_k^\dag T_{k,\ell} \bW_\ell
\end{align*}
where
\[ \bW_k \triangleq e^{-i\eta\bR_k^\dag} \, \bV_k \]
and
\[ T_{k,\ell} \triangleq e^{\eta^2 [\bR_\ell^\dag, \bR_k]} - 1 \]
are scalar coefficients
forming the entries of the Hermitian circulant matrix:
{\small
$$
T=\begin{pmatrix}
     -1+ e^{-\eta^2s}  & -1+ e^{-i \eta^2c} &  -1+ e^{\eta^2s}   &  -1+ e^{i\eta^2c} \\
     -1+ e^{i \eta^2c} & -1+ e^{-\eta^2s}   & -1+ e^{-i \eta^2c} &   -1+ e^{\eta^2s} \\
     -1+ e^{\eta^2s}   & -1+ e^{i \eta^2c}   & -1+ e^{-\eta^2s}  & -1+ e^{-i \eta^2c} \\
     -1+ e^{-i\eta^2c} &   -1+ e^{\eta^2s}   & -1+ e^{i \eta^2c} & -1+ e^{-\eta^2s}   \\
  \end{pmatrix}
$$
}
with $s=\sinh(2\varepsilon)$ and $c=\cosh(2\varepsilon)$.
This matrix  admits the spectral decomposition
$T= \sum_k \lambda_k w_k^\dag w_k$
where
{\small
\begin{align*}
 & w_1 =\begin{pmatrix} \tfrac{1}{2} &  \tfrac{-1}{2} &  \tfrac{1}{2} & \tfrac{-1}{2} \end{pmatrix}
                                    \text{ with } \lambda_1=2(\cosh(\eta^2 s)-\cos(\eta^2c))   \\
 & w_2 =\begin{pmatrix} \tfrac{1}{2} &  \tfrac{1}{2} &  \tfrac{1}{2} & \tfrac{1}{2} \end{pmatrix}
                                    \text{ with } \lambda_2=2(\cosh(\eta^2 s)+\cos(\eta^2c)-2)   \\
 & w_3 =\begin{pmatrix} \tfrac{1}{2} &  \tfrac{-i}{2} &  \tfrac{-1}{2} & \tfrac{i}{2} \end{pmatrix}
                                    \text{ with } \lambda_3=-2(\sinh(\eta^2 s)-\sin(\eta^2c)) \\
 & w_4 =\begin{pmatrix} \tfrac{1}{2} &  \tfrac{i}{2} &  \tfrac{-1}{2} & \tfrac{-i}{2} \end{pmatrix}
                                    \text{ with } \lambda_4=-2(\sinh(\eta^2 s)+\sin(\eta^2c))                                 .
 \end{align*}
 }
Simple numerical  computations % illustrated on figure~\ref{fig:Lambda}
 show that for $\eta=2\sqrt{\pi}$ and  $\eta\varepsilon\in(0,1/2]$ one has
$$
\lambda_4 \leq \lambda_3  \leq 0 \leq \lambda_2 \leq \lambda_1
.
$$
With\begin{align*}
      \bF_1 & = \tfrac{1}{2} \left( \bW_1-\bW_2 + \bW_3 -\bW_4\right)  \\
      \bF_2 & = \tfrac{1}{2} \left( \bW_1+\bW_2 + \bW_3 +\bW_4\right)  \\
      \bF_3 & = \tfrac{1}{2} \left( \bW_1-i\,\bW_2 -\bW_3 +i\,\bW_4\right)  \\
      \bF_4 & = \tfrac{1}{2} \left( \bW_1+i\,\bW_2 -\bW_3 -i\, \bW_4\right)  \\
    \end{align*}
     we have
$$
  \dotex \bW = \sum_k \lambda_k \bF_k^\dag \bF_k
  \leq \lambda_1 (\bF_1^\dag \bF_1 + \bF_2^\dag \bF_2) + \lambda_3 (\bF_3^\dag \bF_3 + \bF_4^\dag \bF_4)
  .
$$
With
\begin{multline*}
 \bF_1^\dag \bF_1 + \bF_2^\dag \bF_2= \tfrac{1}{2} \Big( (\bW_1+\bW_3)^\dag  (\bW_1+\bW_3) + (\bW_2+\bW_4)^\dag  (\bW_2+\bW_4)  \Big)
\end{multline*}
and
\begin{multline*}
\bF_3^\dag \bF_3 + \bF_4^\dag \bF_4= \tfrac{1}{2} \Big( (\bW_1-\bW_3)^\dag  (\bW_1-\bW_3)
 + (\bW_2-\bW_4)^\dag  (\bW_2-\bW_4)  \Big)
\end{multline*}
one gets
{\small
\begin{align*}
	\dotex \bW \leq&
  \tfrac{\lambda_1}{2} (\bW_1+\bW_3)^\dag  (\bW_1+\bW_3)+ \tfrac{\lambda_3}{2}(\bW_1-\bW_3)^\dag  (\bW_1-\bW_3)
  \\
	+&
   \tfrac{\lambda_1}{2} (\bW_2+\bW_4)^\dag  (\bW_2+\bW_4)+ \tfrac{\lambda_3}{2}(\bW_2-\bW_4)^\dag  (\bW_2-\bW_4)
  .
\end{align*}
}
Writing $\bV_k$ as
\[
	\bV_k
	= e^{i\eta\bR_k}-\bI
	= e^{i\tfrac{\bR_k}2}\left( e^{i\eta\tfrac{\bR_k}2} - e^{-i\eta\tfrac{\bR_k}2}\right)
\]
we have
\begin{align*}
  \bW_1+\bW_3  = \left(e^{-i\eta\bR^\dag} e^{i\eta\tfrac{\bR}{2}} -  e^{i\eta\bR^\dag}  e^{-i\eta\tfrac{\bR}{2}} \right) (e^{i \eta\tfrac{\bR}{2}}-e^{-i \eta\tfrac{\bR}{2}})
  \\
    \bW_1-\bW_3  = \left(e^{-i\eta\bR^\dag} e^{i\eta\tfrac{\bR}{2}} + e^{i\eta\bR^\dag}  e^{-i\eta\tfrac{\bR}{2}} \right) (e^{i \eta\tfrac{\bR}{2}}-e^{-i \eta\tfrac{\bR}{2}})
.
\end{align*}
Setting $\Lambda_\pm =e^{-i\eta\bR^\dag} e^{i\eta\tfrac{\bR}{2}} \pm  e^{i\eta\bR^\dag}  e^{-i\eta\tfrac{\bR}{2}}$,
usual computations based on Glauber identity yield ($s= \sinh(2\varepsilon)$)
\[
	\Lambda_\pm^\dagger \Lambda_\pm =  2 e^{-\eta^2 s/8} \left(
  \cosh\big(3 \eta \sinh(\varepsilon) \bP \big) \pm e^{-3\eta^2s/4} \cos\big(\eta\cosh(\varepsilon) \bQ\big) \right).
\]
We thus have
\begin{multline*}
 \tfrac{\lambda_1}{2} (\bW_1+\bW_3)^\dag  (\bW_1+\bW_3)+ \tfrac{\lambda_3}{2}(\bW_1-\bW_3)^\dag  (\bW_1-\bW_3)
 \\ =
  e^{-\eta^2 s /8}  \big(e^{i \eta\tfrac{\bR}{2}}-e^{-i \eta\tfrac{\bR}{2}}\big)^\dag
  \Bigg(   (\lambda_1+\lambda_3) \cosh\big(3 \eta \sinh(\varepsilon) \bP \big)
   \\ +
 (-\lambda_1+\lambda_3)  e^{-3\eta^2 s /4} \cos\big(\eta\cosh(\varepsilon) \bQ\big)
   \Bigg)
   \big(e^{i \eta\tfrac{\bR}{2}}-e^{-i \eta\tfrac{\bR}{2}}\big)
   .
\end{multline*}
With $\lambda_3-\lambda_1 \leq 0$, $\varepsilon>0$  and lemma~\ref{lem:OperatorInequality}, we have
\begin{multline*}
(-\lambda_1+\lambda_3)  e^{-3\eta^2 s/4} \cos\big(\eta\cosh(\varepsilon) \bQ\big)
 \leq (\lambda_1-\lambda_3) e^{-3\eta^2 s/2}\cosh\big(3 \eta \sinh(\varepsilon) \bP \big)
.
\end{multline*}
Consequently
\begin{multline*}
 \tfrac{\lambda_1}{2} (\bW_1+\bW_3)^\dag  (\bW_1+\bW_3)+ \tfrac{\lambda_3}{2}(\bW_1-\bW_3)^\dag  (\bW_1-\bW_3)
 \\ \leq -2  e^{-\eta^2 s/8} \kappa(\eta,\varepsilon)
   \big(e^{i \eta\tfrac{\bR}{2}}-e^{-i \eta\tfrac{\bR}{2}}\big)^\dag
   \, \cosh\big(3 \eta \sinh(\varepsilon) \bP \big) \,
   \big(e^{i \eta\tfrac{\bR}{2}}-e^{-i \eta\tfrac{\bR}{2}}\big)
   .
\end{multline*}
with
\begin{align*}
	\kappa(\eta,\varepsilon)&= -\tfrac{1}{2} \lambda_3 (1- e^{-\frac{3\eta^2 s}{2}})
				-\tfrac12 \lambda_1  (1+ e^{-\frac{3\eta^2 s}{2}}\big)\\
				&=\big(\sinh(\eta^2 s)-\sin(\eta^2c)\big) \big(1- e^{-3\eta^2 s/2}\big)
		-\big(\cosh(\eta^2 s)-\cos(\eta^2c)\big)  \big(1+ e^{-3\eta^2 s/2}\big).
\end{align*}
Similarly,  we have
\begin{multline*}
 \tfrac{\lambda_1}{2} (\bW_2+\bW_4)^\dag  (\bW_2+\bW_4)+ \tfrac{\lambda_3}{2}(\bW_2-\bW_4)^\dag  (\bW_2-\bW_4)
 \\ \leq -2  e^{-\eta^2 s /8} \kappa(\eta,\varepsilon)
   \big(e^{i \eta\bS/2}-e^{-i \eta\bS/2}\big)^\dag
  \, \cosh\big(3 \eta \sinh(\varepsilon) \bQ \big) \,
   \big(e^{i \eta\bS/2}-e^{-i \eta\bS/2}\big)
   .
\end{multline*}
We have also
\begin{align*}
	\bW = 2&  e^{-\eta^2 s/8} %
  \big(e^{i \eta\tfrac{\bR}{2}}-e^{-i \eta\tfrac{\bR}{2}}\big)^\dag %
  \cosh\big( \eta \sinh(\varepsilon) \bP \big) %
   \big(e^{i \eta\tfrac{\bR}{2}}-e^{-i \eta\tfrac{\bR}{2}}\big) %
   \\
	 + &2  e^{-\eta^2 s/8} %
 \big(e^{i \eta\tfrac{\bS}{2}}-e^{-i \eta\tfrac{\bS}{2}}\big)^\dag %
  \cosh\big( \eta \sinh(\varepsilon) \bQ \big) %
   \big(e^{i \eta\tfrac{\bS}{2}}-e^{-i \eta\tfrac{\bS}{2}}\big) %
   . %
\end{align*}
Since   $\kappa(\eta,\varepsilon) >0$ we have
$
\dotex \bW \leq  -\kappa(\eta,\varepsilon)  \bW
$
using  $\cosh(3\sinh(\varepsilon)\bP) \geq \cosh(\sinh(\varepsilon)\bP)$ and $\cosh(3\sinh(\varepsilon)\bQ) \geq \cosh(\sinh(\varepsilon)\bQ)$.

\section{An operator inequality}
\begin{lemma}\label{lem:OperatorInequality}
  Take two operators Hermitian  $\bQ$ and $\bP$ such that $[\bQ,\bP]=i \bI$. Then
	\begin{equation*}
		 \forall \eta,\varepsilon \in\mathbb{R} \quad \quad
e^{-\tfrac{3\eta^2 |\sinh(2\varepsilon)|}{4}} \cosh\big(3 \eta \sinh(\varepsilon) \bP \big)
 \geq  \pm  \cos\big(\eta\cosh(\varepsilon) \bQ\big)
.
	\end{equation*}
\end{lemma}
\begin{proof}
%{\small
  Set $\bR= \cosh(\varepsilon) \bQ + i \sinh(\varepsilon)\bP$
and $\Lambda_\pm =e^{-i\eta\bR^\dag} e^{i\eta\tfrac{\bR}{2}} \pm  e^{i\eta\bR^\dag}  e^{-i\eta\tfrac{\bR}{2}}$,
	then usual computations based on Glauber identity yield
	\[
	\Lambda_\pm^\dagger \Lambda_\pm =  2 e^{-\eta^2 s/8} \left(
  \cosh\big(3 \eta \sinh(\varepsilon) \bP \big) \pm e^{-3\eta^2s/4} \cos\big(\eta\cosh(\varepsilon) \bQ\big) \right)
\]
with $s=\sinh(2\varepsilon)$.
Thus,  for any $\eta$ and $\varepsilon$, the operators
$$ \cosh\big(3 \eta \sinh(\varepsilon) \bP \big) \pm  e^{-3\eta^2s/4} \cos\big(\eta\cosh(\varepsilon) \bQ\big)
$$ are non-negative. This means that
$$
 e^{3\eta^2 s/4} \cosh\big(3 \eta \sinh(\varepsilon) \bP \big)\geq  \pm  \cos\big(\eta\cosh(\varepsilon) \bQ\big)
.
$$
We conclude by changing   $\varepsilon$ to $-\varepsilon$.
%}
\end{proof}

\end{document}